\documentclass[a4paper,10pt]{article}

\usepackage[utf8]{inputenc} % allow utf-8 input
\usepackage[T1]{fontenc}    % use 8-bit T1 fonts
\usepackage{hyperref}       % hyperlinks
\usepackage{url}            % simple URL typesetting
\usepackage{booktabs}       % professional-quality tables
\usepackage{nicefrac}       % compact symbols for 1/2, etc.
\usepackage{microtype}      % microtypography
\usepackage{graphicx}
\usepackage{mathtools}

\usepackage{amsmath}
\usepackage{fourier}
\usepackage{braket}
\usepackage[standard,amsmath,thmmarks]{ntheorem}

\usepackage{algorithm}
\usepackage{algpseudocode} 

\usepackage{bbold}

\newcommand{\du}{\mathrm{d}}

\DeclareMathOperator{\dom}{\mathrm{dom}}

\DeclareMathOperator*{\argmin}{ \arg \min }

\DeclareMathOperator{\tr}{\mathrm{Tr}}

\usepackage{authblk}
\usepackage{xcolor}
\usepackage{interval}

\title{Maximum-Likelihood Quantum State Tomography by Cover's Method with Non-Asymptotic Analysis}
\author[1]{Chien-Ming Lin}
\author[2,3,4]{Hao-Chung Cheng}
\author[1,3]{Yen-Huan Li}

\date{\empty}

\affil[1]{Department of Computer Science and Information Engineering, National Taiwan University}
\affil[2]{Department of Electrical Engineering, National Taiwan University}
\affil[3]{Department of Mathematics, National Taiwan University}
\affil[4]{Hon Hai (Foxconn) Quantum Computing Centre}

\begin{document}

\maketitle

\begin{abstract}
We propose an iterative algorithm that computes the maximum-likelihood estimate in quantum state tomography. 
The optimization error of the algorithm converges to zero at an \( O ( ( 1 / k ) \log D ) \) rate, where \( k \) denotes the number of iterations and \( D \) denotes the dimension of the quantum state. 
The per-iteration computational complexity of the algorithm is \( O ( D ^ 3 + N D ^2 ) \), where \( N \) denotes the number of measurement outcomes. 
The algorithm can be considered as a parameter-free correction of the \( R \rho R \) method [A. I. Lvovsky. Iterative maximum-likelihood reconstruction in quantum homodyne tomography. \textit{J. Opt. B: Quantum Semiclass. Opt.} 2004] [G. Molina-Terriza et al. Triggered qutrits for quantum communication protocols. \textit{Phys. Rev. Lett.} 2004.]. 
\end{abstract}

\section{Introduction}

Quantum state tomography aims to estimate the quantum state of a physical system, given measurement outcomes (see, e.g., \cite{Paris2004} for a complete survey). 
There are various approaches to quantum state tomography, such as trace regression \cite{Flammia2012,Gross2010,Opatrny1997,Yang2020,Youssry2019}, maximum-likelihood estimation \cite{Hradil1997,Hradil2004}, Bayesian estimation \cite{Blume-Kohout2010,Blume-Kohout2010a}, and recently proposed deep learning-based methods \cite{Ahmed2020,Quek2021}\footnote{A  confusion the authors frequently encounter is that many people mix state tomography with the notion of shadow tomography introduced by Aaronson \cite{Aaronson2018,Aaronson2020}.
State tomography aims at estimating the quantum state, whereas shadow tomography aims at estimating the probability distribution of measurement outcomes.
Indeed, one interesting conclusion by Aaronson is that shadow tomography requires much less data than state tomography. }. 
Among existing approaches, the maximum-likelihood estimation approach has been standard in practice and enjoys favorable asymptotic statistical guarantees (see, e.g., \cite{Hradil2004,Scholten2018}). 
The maximum-likelihood estimator is given by the optimization problem: 
\begin{equation}
\hat{\rho} \in \argmin_{\rho \in \mathcal{D}} f ( \rho ) , \quad f ( \rho ) \coloneqq \frac{1}{N} \sum_{n = 1}^N - \log \tr ( M_n \rho ) , \label{eq_problem}
\end{equation}
for some Hermitian positive semi-definite matrices \( M_n \), where \( \mathcal{D} \) denotes the set of quantum density matrices, i.e., 
\[
\mathcal{D} \coloneqq \Set{ \rho \in \mathbb{C}^{D \times D} | \rho = \rho^\dagger, \rho \geq 0, \tr ( \rho ) = 1 } ,  
\]
and \( N \) denotes the number of measurement outcomes. 
We write \( \rho^\dagger \) for the conjugate transpose of \( \rho \). 

\( R \rho R \) is a numerical method developed to solve \eqref{eq_problem} \cite{Lvovsky2004,MolinaTerriza2004}. 
Given a positive definite \( \rho_1 \in \mathcal{D} \), \( R \rho R \) iterates as 
\[
\rho_{k + 1} = \mathcal{N} ( R ( \rho_k ) \rho_k R ( \rho_k ) ), \quad R ( \rho_k ) \coloneqq - \nabla f ( \rho_k ) = \frac{1}{N} \sum_{n = 1}^N \frac{M_n}{\tr ( M_n \rho_k )} , \quad \forall k \in \mathbb{N} , 
\]
where the mapping \( \mathcal{N} \) scales its input such that \( \tr ( \rho_{k + 1} ) = 1 \), and \( \nabla f \) denotes the gradient mapping of \( f \). 
\( R \rho R \) is parameter-free (i.e., it does not require parameter tuning) and typically converges fast in practice. 
Unfortunately, one can construct a synthetic data-set on which \( R \rho R \) does not converge \cite{Rehacek2007}. 

According to \cite{Lvovsky2004}, \( R \rho R \) is inspired by Cover's method\footnote{Indeed, Cover's method coincides with the expectation maximization method for solving \eqref{eq_classical} and hence is typically called expectation maximization in literature. Nevertheless, Cover's and our derivations and convergence analyses do not need and are not covered by existing results on expectation maximization, so we do not call the method expectation maximization to avoid possible confusions.} for solving the optimization problem \cite{Cover1984}: 
\begin{equation}
x^\star \in \argmin_{ x \in \Delta } g ( x ) , \quad g ( x ) \coloneqq \frac{1}{N} \sum_{n = 1}^N - \log \braket{a_n, x} , \label{eq_classical}
\end{equation}
for some entry-wise non-negative vectors \( a_n \in \mathbb{R}^D \), where \( \Delta \) denotes the probability simplex in \( \mathbb{R}^D \), i.e., 
\[
\Delta \coloneqq \Set{ x = ( x_1, \ldots, x_D ) \in \mathbb{R}^D | x_d \geq 0 ~ \forall ~ d , \sum_{d = 1}^D x_d = 1 } ,  
\]
and the inner product is the one associated with the Euclidean norm. 
The optimization problem appears when one wants to compute the \emph{growth-optimal portfolio} for long-term investment (see, e.g., \cite{MacLean2012}). 
Given an entry-wise strictly positive initial iterate \( x_1 \in \mathbb{R}_{++}^D \), Cover's method iterates as 
\[
x_{k + 1} = x_k \circ \left( - \nabla g ( x_k ) \right) , \quad \forall k \in \mathbb{N} , 
\]
where \( \circ \) denotes the entry-wise product, aka the Schur product. 
Cover's method is guaranteed to converge to the optimum \cite{Cover1984}. 
Indeed, if the matrices in \eqref{eq_problem} share the same eigenbasis, then it is easily checked that \eqref{eq_problem} is equivalent to \eqref{eq_classical} but \( R \rho R \) is not equivalent to Cover's method\footnote{Cover's method does not need the scaling mapping \( \mathcal{N} \). One can check that its iterates are already in \( \Delta \).}. 
This explains why \( R \rho R \) does not inherit the convergence guarantee of Cover's method.  

Rehacek et al. proposed a diluted correction of \( R \rho R \) \cite{Rehacek2007}.
Given a positive definite initial iterate \( \rho_1 \in \mathcal{D} \), diluted \( R \rho R \) iterates as 
\[
\rho_{k + 1} = \mathcal{N} ( \left( R ( \rho_k ) + \alpha_k I \right) \rho_k \left( R ( \rho_k ) + \alpha_k I \right) ), 
\]
where the parameter \( \alpha_k \) is chosen by exact line search. 
Later, Goncalves et al. proposed a variant of diluted \( R \rho R \) by adopting Armijo line search \cite{Goncalves2014}. 
Both versions of diluted \( R \rho R \) are guaranteed to converge to the optimum. 
Unfortunately, the convergence guarantees of both versions of diluted \( R \rho R \) are asymptotic and do not allow us to characterize the \emph{iteration complexity}, the number of iterations required to obtain an approximate solution of \eqref{eq_problem}. 
In particular, the dimension \( D \) grows exponentially with the number of qubits and can be huge, but the dependence of the iteration complexities of the diluted \( R \rho R \) methods on \( D \) is unclear. 

We propose the following algorithm. 
\begin{itemize}
\item Set \( \rho_1 = I / D \), where \( I \) denotes the identity matrix. 
\item For each \( k \in \mathbb{N} \), compute 
\[
\rho_{k + 1} = \mathcal{N} \left( \exp \left( \log ( \rho_k ) + \log ( R ( \rho_k ) ) \right) \right) , 
\]
where \( \exp \) and \( \log \) denote matrix exponential and logarithm, respectively. 
\end{itemize}
Notice that the objective function implicitly requires that \( \tr ( M_n \rho_k ) \) are strictly positive; 
otherwise, \( R ( \rho_k ) \) does not exist as \( - \log \tr ( M_n \rho_k ) \) is not well-defined. 
Our initialization and iteration rule guarantee that \( \rho_k \) are full-rank and \( \tr ( M_n \rho_k ) \) are strictly positive. 

Let us discuss the computational complexity of the proposed algorithm. 
The computational complexity of computing \( \nabla f ( \rho_k ) \) is \( O ( N D ^ 2 ) \). 
The computational complexities of computing matrix logarithm and exponential are \( O ( D ^ 3 ) \). 
The per-iteration computational complexity is hence \( O ( D ^ 3 + N D ^ 2 ) \). 

We can observe that the proposed algorithm recovers Cover's method when \( \rho_k \) and \( \nabla f ( \rho_k ) \) share the same eigenbasis. 
We show that the proposed algorithm indeed converges and its iteration complexity is logarithmic in the dimension \( D \). 

\begin{theorem} \label{thm_main}
Assume that \( \bigcap_{n = 1}^N \ker ( M_n ) = \set{ 0 } \). 
Let \( ( \rho_k )_{k \in \mathbb{N}} \) be the sequence of iterates generated by the proposed method. 
Define \( \overline{\rho}_k \coloneqq ( \rho_1 + \cdots + \rho_k ) / k \). 
Then, for every \( \varepsilon > 0 \), we have \( f ( \overline{\rho}_k ) - f ( \hat{\rho} ) \leq \varepsilon \) if \( k \geq ( 1 / \varepsilon ) \log D \). 
\end{theorem}

\begin{remark}
Suppose \( \mathcal{K} \coloneqq \bigcap_{n = 1}^N \ker ( M_n ) \neq \set{ 0 } \). 
Let \( U \) be a matrix whose columns form an orthogonal basis of \( \mathcal{K}^\perp \). 
Then, it suffices to solve \eqref{eq_problem} on a lower-dimensional space by replacing \( A_n \) with \( U^\dagger A_n U \) in the objective function. 
\end{remark}

Recall that \eqref{eq_classical} and Cover's method are special cases of \eqref{eq_problem} and the proposed algorithm, respectively. 
Moreover, Cover's method is equivalent to the expectation maximization method for Poisson inverse problems \cite{Vardi1993}. 
Theorem \ref{thm_main} is hence of independent interest even for computing the growth-optimal portfolio by Cover's method and solving Poisson inverse problems by expectation maximization, showing that the iteration complexities of both are also \( O ( \varepsilon^{-1} \log D ) \). 
This supplements the asymptotic convergence results in \cite{Cover1984} and \cite{Vardi1985}. 
Whereas the same iteration complexity bound for Cover's method in growth-optimal portfolio is immediate from a lemma due to Iusem \cite[Lemma 2.2]{Iusem1992}, it is currently unclear to us how to extend Iusem's analysis to the quantum setup. 

%There are some other methods that compute the maximum-likelihood estimator with non-asymptotic iteration complexity guarantees. 
%Proximal Newton and gradient methods \cite{Tran-Dinh2015b}, several versions of the Frank-Wolfe method \cite{Carderera2021,Dvurechensky2020,Odor2016,Zhao2020}, NoLips \cite{Bauschke2017}, the primal-dual hybrid gradient method \cite{Chambolle2011,Chambolle2018}, and mirror prox \cite{He2019} all converge at a \( O ( 1 / k ) \) rate as our proposed algorithm. 
%However, existing methods have at least one of the following disadvantages:
%\begin{itemize}
%\item Unclear dependence of the convergence rates on the dimension \( D \) \cite{Bauschke2017,Carderera2021,Dvurechensky2020,Odor2016}, 
%\item Undesirable slowing down with a large \( N \) \cite{Bauschke2017,Tran-Dinh2015b}, and 
%\item Computationally expensive iteration rules \cite{Bauschke2017,Dvurechensky2020,Tran-Dinh2015b,Zhao2020}. 
%\end{itemize} 
%A recent stochastic method proposed by the authors has a much slower \( O ( \sqrt{ ( D / k ) \log D } ) \) convergence rate, though its per-iteration computational complexity is independent of \( N \) and hence favourable when \( N \) is large \cite{Lin2020}. 
%Perhaps surprising to those unfamiliar with optimization theory, the projected gradient descent does not have a convergence guarantee and may not be well-defined for maximum-likelihood quantum state tomography;
%the interested reader is referred to the discussions in, e.g., \cite{Knee2018a,You2021}. 

\section{Proof of Theorem \ref{thm_main}}

For convenience, let \( M \) be a random matrix following the empirical probability distribution of \( M_1, \ldots, M_N \)\footnote{Notice the derivation here is not restricted to the empirical probability distribution.}. 
Then, we have \( f ( \rho ) = \mathsf{E} \left[ - \log \braket{ M, \rho } \right] \), where \( \mathsf{E} \) denotes the mathematical expectation. 
Define 
\[
r ( \rho ) \coloneqq \frac{M}{\braket{ M, \rho }} , \quad R ( \rho ) \coloneqq - \nabla f ( \rho ) = \mathsf{E} \left[ r ( \rho ) \right] , 
\]
where the inner product, as well as all inner products in the rest of this section, is the Hilbert-Schmidt inner product. 
We start with an error upper bound. 

\begin{lemma} \label{lem_error}
For any density matrix \( \rho \) such that \( R ( \rho ) \) exists, 
\[
f ( \rho ) - f ( \hat{\rho} ) \leq \max_{\sigma \in \mathcal{D}} \braket{ \log R ( \rho ), \sigma } . 
\] 
\end{lemma}

\begin{remark}
Notice that \( R ( \rho ) \) is always positive definite and \( \log R ( \rho ) \) is always well-defined. 
Otherwise, suppose there exists some vector \( u \) such that \( \braket{ u, R ( \rho ) u } = 0 \). 
Then, as \( M_n \) are positive semi-definite, we have \( \braket{ u, M_n u } = 0 \) for all \( n \), violating the assumption that \( \bigcap_{n = 1}^N \ker ( M_n ) = \set{ 0 } \). 
\end{remark}

\begin{proof}[Lemma \ref{lem_error}]
By Jensen's inequality, we write 
\begin{align*}
f ( \rho ) - f ( \hat{\rho} ) & = \mathsf{E} \left[ \log \Braket{ \frac{M}{\braket{ M, \rho }}, \hat{\rho} } \right] \\
& \leq \log \left[ \Braket{ \mathsf{E} \left[ \frac{M}{\braket{ M, \rho }} \right], \hat{\rho} } \right] \\
& = \log \braket{ R ( \rho ), \hat{\rho} } \\
& \leq \log \lambda_{\max} ( R ( \rho ) ) \\
& = \lambda_{\max} ( \log R ( \rho ) ) \\
& = \max_{\sigma \in \mathcal{D}} \braket{ \log R ( \rho ), \sigma } .   
\end{align*}
\end{proof}

Deriving the following lemma is the major technical challenge in our convergence analysis. 
The lemma shows that the mapping \( \log R ( \cdot ) \) is operator convex. 

\begin{lemma} \label{lem_convexity}
For any density matrix \( \sigma \), the function \( \varphi ( \rho ) \coloneqq \braket{ \log R ( \rho ), \sigma } \) is convex. 
\end{lemma}

\begin{proof}
Equivalently, we want to show that \( D^2 \varphi ( \rho ) [ \delta, \delta ] \geq 0 \) for all \( \rho \in \dom D^2 \varphi \) and Hermitian \( \delta \in \mathbb{C}^{D \times D} \).
Define \( A_0 \coloneqq R ( \rho ) \).  
By \cite[Example 3.22 and Exercise 3.24]{Hiai2014}, we have 
\begin{align*}
& A_1 \coloneqq D R ( \rho ) [ \delta ] = - \mathsf{E} \left[ r ( \rho ) \braket{ r ( \rho ), \delta } \right] , \\
& A_2 \coloneqq D^2 R ( \rho ) [ \delta, \delta ] = 2 \mathsf{E} \left[ r ( \rho ) \braket{ r ( \rho ), \delta } ^ 2 \right] , \\
& D \log ( A_0 ) [ A_2 ] = \int_0^\infty ( A_0 + s I )^{-1} A_2 ( A_0 + s I )^{-1} \, \du s , \\
& D ^ 2 \log ( A_0 ) [ A_1, A_1 ] = - 2 \int_0^\infty ( A_0 + s I )^{-1} A_1 ( A_0 + s I )^{-1} A_1 ( A_0 + s I )^{-1} \, \du s . 
\end{align*}
Define \( \Phi ( \rho ) \coloneqq \log R ( \rho ) \). 
Recall the chain rule for the second-order Fr\'{e}chet derivative (see, e.g., \cite[p. 316]{Bhatia1997}): 
\[
D ^ 2 \Phi ( \rho ) [ \delta, \delta ] = D ^ 2 \log ( R ( \rho ) ) [ D R ( \rho ) [ \delta ], D R ( \rho ) [ \delta ] ] + D \log ( R ( \rho ) ) [ D ^ 2 R ( \rho ) [ \delta, \delta ] ] . 
\]
Then, we write 
\begin{align*}
D^2 \varphi ( \rho ) [ \delta, \delta ] & = \tr \left( \sigma D^2 \Phi ( \rho ) [ \delta, \delta ] \right) \\
& = -2 \int_0^\infty \tr \left( \sigma ( A_0 + s I )^{-1} A_1 ( A_0 + s I )^{-1} A_1 ( A_0 + s I )^{-1} \right) \, \du s + \\
& \quad \,\, \int_0^\infty \tr ( \sigma ( A_0 + s I )^{-1} A_2 ( A_0 + s I )^{-1} ) \, \du s \\
& = 2 \int_0^\infty \tr \left( B_s \left( ( A_0 + s I )^{-1} \sigma ( A_0 + s I )^{-1} \right) \right) \, \du s , 
\end{align*}
where 
\[
B_s \coloneqq \frac{A_2}{2} - A_1 ( A_0 + s I )^{-1} A_1 . 
\]
Since \( ( A_0 + s I )^{-1} \sigma ( A_0 + s I )^{-1} \) is obviously positive semi-definite, it suffices to show that \( B_s \) is positive semi-definite for all \( s \geq 0 \).
We write 
\begin{align*}
B_s & \geq \frac{A_2}{2} - A_1 A_0^{-1} A_1 \\
& = \mathsf{E} \left[ r ( \rho ) \braket{ r ( \rho ), \delta } ^ 2 \right] - \mathsf{E} \left[ r ( \rho ) \braket{ r ( \rho ), \delta } \right] \left( \mathsf{E} \left[ r ( \rho ) \right] \right)^{-1} \mathsf{E} \left[ r ( \rho ) \braket{ r ( \rho ), \delta } \right] , 
\end{align*} 
which is positive semi-definite by an extension of the Cauchy-Schwarz inequality due to Lavergne \cite{Lavergne2008}\footnote{Whereas Lavergne considers the real matrix case, we notice the proof directly extends to the Hermitian matrix case.}. 
\end{proof}

Now, we are ready to prove Theorem \ref{thm_main}. 

\begin{proof}[Theorem \ref{thm_main}]
By the Golden-Thompson inequality, we write 
\[
\tr \left( \exp \left( \log ( \rho_k ) + \log ( - \nabla f ( \rho_k ) ) \right) \right) \leq \tr \left( \rho_k \left( - \nabla f ( \rho_k ) \right) \right) = \mathsf{E} \left[ \frac{\braket{ M, \rho_k }}{ \braket{ M, \rho_k } } \right] = 1 . 
\]
Therefore, by the iteration rule of the proposed algorithm and operator motononicity of the matrix logarithm, we have 
\begin{equation}
\log \rho_{k + 1} \geq \log \rho_k + \log R ( \rho_k ) . \label{eq_iterInLog}
\end{equation}
Then, for any \( K \in \mathbb{N} \), we write 
\begin{align*}
f ( \overline{\rho}_K ) - f ( \hat{\rho} ) & \leq \max_{\sigma \in \mathcal{D}} \tr \left( \sigma \log R \left( \overline{\rho}_K \right) \right) \\
& \leq \max_{\sigma \in \mathcal{D}} \frac{1}{K} \sum_{k = 1}^K \tr \left( \sigma \log R \left( \rho_k \right) \right) \\
& \leq \frac{1}{K} \max_{\sigma \in \mathcal{D}} \sum_{k = 1}^K \tr \left( \sigma \left( \log ( \rho_{k + 1} ) - \log ( \rho_k ) \right) \right) \\
& = \frac{1}{K} \max_{\sigma \in \mathcal{D}} \tr \left( \sigma \left( \log ( \rho_{K + 1} ) - \log ( \rho_1 ) \right) \right) \\
& \leq \frac{\log D}{K} , 
\end{align*}
where the first inequality follows from Lemma \ref{lem_error}, the second follows from Lemma \ref{lem_convexity} and Jensen's inequality, the third follows from \eqref{eq_iterInLog}, and the last follows from the fact that \( \log ( \rho_{K + 1} ) \leq 0 \). 
\end{proof}

\section*{Acknowledgement}
C.-M. Lin and Y.-H. Li are supported by the Young Scholar Fellowship (Einstein Program) of the Ministry of Science and Technology of Taiwan under grant numbers MOST MOST 109-2636-E-002-025 and MOST 110-2636-E-002-012.
H.-C.~Cheng is supported by the Young Scholar Fellowship (Einstein Program) of the Ministry of Science and Technology in Taiwan (R.O.C.) under grant number MOST 109-2636-E-002-001 \& 110-2636-E-002-009, and is supported by the Yushan Young Scholar Program of the Ministry of Education in Taiwan (R.O.C.) under grant number NTU-109V0904 \& NTU-110V0904.

%\bibliographystyle{../../my_ref_3/_bib/halpha}
%\bibliography{../../my_ref_3/_bib/list}
\bibliographystyle{halpha}
\bibliography{list}

\end{document}